\documentclass[12pt]{article}
\usepackage{amsfonts}
\usepackage{enumerate}
\usepackage{amsmath}
\usepackage{amsthm}
\usepackage{color}
\usepackage{addfont}
\addfont{OT1}{rsfs10}{\rsfs}

\def\1{{\bf 1}}

\begin{document}

\newtheorem{thm}{Theorem}[section]
\newtheorem{lem}[thm]{Lemma}
\newtheorem{prop}[thm]{Proposition}
\newtheorem{cor}[thm]{Corollary}
\newtheorem{defi}[thm]{Definition}
\newtheorem{remark}[thm]{Remark}
\newtheorem{result}[thm]{Result}
\newtheorem{exm}[thm]{Example}
\title
{\bf Type IV-II codes over $\mathbb{Z}_4$ constructed from generalized bent functions}
\author
{Sara Ban (sban@math.uniri.hr)\\
[3pt]
Sanja Rukavina (sanjar@math.uniri.hr)\\
{\it\small Department of Mathematics}\\
{\it\small University of Rijeka, 51000 Rijeka, Croatia }\\[-15pt]
}
\date{}
\maketitle

\begin{abstract}
\noindent
A Type IV-II  $\mathbb{Z}_4$-code is a self-dual code over  $\mathbb{Z}_4$ with the property that all Euclidean weights are divisible by eight and
all codewords have even Hamming weight. In  this paper we use generalized bent functions for a construction of self-orthogonal codes over $\mathbb{Z}_4$ of length $2^m$, for $m$ odd, $m\geq3$, and prove that for  $m\geq5$ those codes can be extended to Type IV-II $\mathbb{Z}_4$-codes. From that family of Type IV-II $\mathbb{Z}_4$-codes, we obtain a family of self-dual Type II binary codes by using Gray map.
We also consider the weight distributions of the obtained codes and the structure of the supports of the minimum weight codewords.

\end{abstract}

\section{Introduction}

According to \cite{4dec}, the first paper in english on bent functions has been written in 1966 by Rothaus, but its final version was published ten years later in \cite{Roth}. 
Since then, bent functions have been a subject of interest of many researchers (see \cite{4dec}).
Among other things, relationships between bent functions and codes have been intensively studied. 
For example, cyclic codes and their connection with hyper-bent and bent functions are explored in \cite{hyper},
a construction of linear codes with two or three weights from weakly regular bent functions is given in \cite{weakly}, 
and in \cite{vectorial} bent vectorial functions are used for a construction of a two-parameter family of binary linear codes that do not satisfy
the conditions of the Assmus-Mattson theorem, but nevertheless hold 2-designs, and a new coding-theoretic characterization of bent vectorial functions is presented.\
Trace codes over $\mathbb{Z}_4$ based on Boolean functions and their supports are explored in \cite{hal}, and three-weight codes are obtained from bent and semi-bent functions.\\
Generalized bent functions were  introduced in \cite{Kumar}.
In \cite{Sch}, Schmidt considered generalized bent functions for a construction of constant-amplitude codes over $\mathbb{Z}_4$ of length $2^m$.
In  this paper we use generalized bent functions for a construction of a family of Type IV-II codes over $\mathbb{Z}_4$ of length $2^m$, for $m$ odd, $m\geq5$.
As a consequence, by Gray map, we obtain a family of self-dual Type II binary codes of length  $2^{m+1}$, for $m$ odd, $m\geq5$. 
Further, we consider the weight distributions of the obtained codes and the supports of the minimum weight codewords.\\
This paper is organized as follows. Section \ref{pre} gives definitions and basic properties of codes over $\mathbb{Z}_4$ and generalized bent functions.
In Section \ref{constr} the construction of Type II codes over $\mathbb{Z}_4$  of length $2^m$, for $m$ odd, $m\geq3$, from generalized bent functions is introduced.
We prove that for $m\geq5$ the constructed $\mathbb{Z}_4$-codes are also of Type IV. We give the Euclidean weight distribution, the Lee weight distribution and the symmetrized weight enumerator for the constructed codes.
By Gray map, we obtain a family of self-dual Type II binary codes of length  $2^{m+1}$, for $m$ odd, $m\geq5$. 
Finally, in Section \ref{struc}, we consider examples for $m=3$ and $m=5$, and observe the minimum weight codewords and the structure of their supports.
For the construction of examples we used Magma \cite{magma}.

\section  {Preliminaries} \label{pre}

We assume that the reader is familiar with the basic facts of coding theory. We refer the reader to \cite{FEC} in relation to terms not defined in this paper.

Let $\mathbb{F}_q$ be the field of order $q$, where $q$ is a prime power.
A code $C$ over $\mathbb{F}_q$ of length $n$ is any subset of $\mathbb{F}_q^n$. A $k$-dimensional subspace of $\mathbb{F}_q^n$ is called an $[n,k]$ {\it $q$-ary linear code}. 
An element of a code is called a {\it codeword}. A {\it generator matrix} for an $[n, k]$ code $C$ is any
$k \times n$ matrix whose rows form a basis for $C.$
A linear code $C$ of length $n$ is {\it cyclic} if for every codeword $(c_0,\dots,c_{n-2},c_{n-1})$ in $C$ the codeword $(c_{n-1},c_0,\dots, c_{n-2})$ is also in $C.$ 

If $q=2$, then the code is called {\it binary}.
The {\it (Hamming) weight} of a codeword $x\in \mathbb{F}_2^n$ is the number of non-zero coordinates in $x.$ If the minimum weight $d$ of an $[n,k]$ binary linear code is known, then we refer to the code as an $[n,k,d]$ binary linear code.
Binary linear codes for which all codewords have even weight  are called {\it (singly) even} and those among them for which all codewords have weight divisible by four are called {\it doubly even}.\\
Let $C$ be a binary linear code of length $n$. The {\it dual code} $C^\bot$ of $C$ is defined as
$$C^\bot=\{x\in \mathbb{F}_2^n\ |\ \left\langle x, y\right\rangle=0\ \text{for all}\ y\in C\},$$
where $\left\langle x,\\ y\right\rangle=x_1y_1+x_2y_2+\dots+x_ny_n\ (\text{mod}\ 2)$ for 
$x=(x_1,x_2,\dots,x_n)$ and $y=(y_1,y_2,\dots,y_n).$ The code $C$ is {\it self-orthogonal} if $C\subseteq C^\bot$, and it is {\it self-dual} if $C = C^\bot$. The dual code of a cyclic code is cyclic. A self-dual doubly even binary code is called a {\it Type II binary code}.

Let $\mathbb{Z}_4$ denote the ring of integers modulo $4.$ A linear code $C$ of length $n$ over $\mathbb{Z}_4$ (i.e., a {\it $\mathbb{Z}_4$-code}) is a $\mathbb{Z}_4$-submodule of $\mathbb{Z}_4^n.$
Two $\mathbb{Z}_4$-codes are {\it (monomially) equivalent} if one can be obtained from the other by permuting the coordinates and (if necessary) changing the signs of certain coordinates. Codes differing by only a permutation of coordinates are called {\it permutation equivalent}.
{\it The permutation automorphism group} of a $\mathbb{Z}_4$-code $C$ is the group of all coordinate permutations that fix $C$ set-wise.
The {\it support } of a codeword $x\in \mathbb{Z}_4^n$ is the set of non-zero positions in $x.$
Denote the number of coordinates $i$ (where $i=0,1,2,3$) in a codeword $x\in \mathbb{Z}_4^n$  by $n_i(x).$
The codeword $x\in \mathbb{Z}_4^n$ is {\it even} if $n_1(x)=n_3(x)=0$.
The {\it Hamming weight} of a codeword $x$ is $wt_H(x) =n_1(x) + n_2(x) + n_3(x),$ the {\it Lee weight} of $x$ is $wt_L (x) = n_1(x) + 2n_2(x) + n_3(x),$ and the
{\it Euclidean weight} of $x$ is $wt_E (x) = n_1(x) + 4n_2(x) + n_3(x).$ It holds that $wt_E(x)\equiv x_1^2+\dots +x_n^2\ (\text{mod}\ 8)$ for every $x\in \mathbb{Z}_4^n.$ 
We will denote by $d_H(C)$, $d_L(C)$ and $d_E(C)$, the minimum Hamming weight, the minimum Lee weight and the minimum Euclidean weight of the code $C$, respectively.
 The {\it symmetrized weight enumerator} of a $\mathbb{Z}_4$-code $C$ is defined as
$$swe_C(a,b,c)=\sum_{x\in C} a^{n_0(x)}b^{n_1(x)+n_3(x)}c^{n_2(x)}.$$

Let $C$ be a $\mathbb{Z}_4$-code of length $n$. The {\it dual code} $C^\bot$ of the code $C$ is defined as
\[
C^\bot=\{x\in \mathbb{Z}_4^n\,|\,\left\langle x, y\right\rangle=0\ \text{for all}\ y\in C\},
\]
where $\left\langle x,\\ y\right\rangle=x_1y_1+x_2y_2+\dots+x_ny_n\ (\text{mod}\ 4)$ for
$x=(x_1,x_2,\dots,x_n)$ and $y=(y_1,y_2,\dots,y_n).$ The code $C$ is {\it self-orthogonal} when $C\subseteq C^\bot$ and {\it self-dual} if $C= C^\bot.$ If $C$ is a self-orthogonal  $\mathbb{Z}_4$-code, then $wt_L(c)$ is even for all $c\in C.$
A self-dual $\mathbb{Z}_4$-code of length $n$ contains exactly $2^n$ codewords.
{\it Type II $\mathbb{Z}_4$-codes} are self-dual $\mathbb{Z}_4$-codes which have the property that all
Euclidean weights are divisible by eight. {\it Type IV $\mathbb{Z}_4$-codes} are self-dual $\mathbb{Z}_4$-codes with all codewords of even
Hamming weight (see, e.g., \cite{typeIV}). A Type IV code that is also Type II is called a {\it Type IV-II $\mathbb{Z}_4$-code}.

Every $\mathbb{Z}_4$-code $C$ contains a set of $k_1+k_2$ codewords $\{c_1,c_2,\dots,c_{k_1},c_{k_1+1}, \dots, \\c_{k_1+k_2}\}$ such that every codeword in $C$ is uniquely expressible in the form
\[
\sum_{i=1}^{k_1}a_ic_i+\sum_{i=k_1+1}^{k_1+k_2}a_ic_i,
\]
where $a_i\in \mathbb{Z}_4$ and $c_i$ has at least one coordinate equal to 1 or 3, for $1\leq i\leq k_1,$
$a_i\in \mathbb{Z}_2$  and $c_i$ has all coordinates equal to 0 or 2, for $k_1+1\leq i\leq k_1+k_2.$
We say that $C$ is of {\it type} $4^{k_1}2^{k_2}.$
The matrix whose rows are  $c_i,\ 1\leq i\leq k_1+k_2,$ is called a {\it generator matrix} for $C.$ A generator matrix $G$ of a $\mathbb{Z}_4$-code $C$ is in {\it standard form} if

\begin{equation}\label{std}
G=\left[\begin{tabular}{ccc}
$I_{k_1}$&$A$&$B_1+2B_2$\\
$O$&$2I_{k_2}$&$2D$\\
\end{tabular}\right],
\end{equation}

where $A, B_1,B_2$ and $D$ are matrices with entries from $\mathbb{Z}_2$, $O$ is the $k_2\times k_1$ null matrix,
and $I_{m}$ denotes the identity matrix of order $m$.
If $C$ is a self-dual $\mathbb{Z}_4$-code of  length $n$,

then $2k_1+k_2=n$ and the matrix $B_1+2B_2$ in $G$ is of order $k_1.$ Any $\mathbb{Z}_4$-code is permutation equivalent to a code with generator matrix in standard form.

Because of $$wt_E(x+y)\equiv wt_E(x)+wt_E(y)+2\left\langle x, y\right\rangle \ (\text{mod}\ 8)$$
for all $x, y \in \mathbb{Z}_4^n,$ 
every self-orthogonal  $\mathbb{Z}_4$-code which has a generator matrix such that
all rows have Euclidean weights divisible by $8$ consists of codewords whose Euclidean weights are divisible by $8$. 

Let $C$ be a $\mathbb{Z}_4$-code of length $n.$  
There are two binary linear codes of length $n$ associated with $C$: the binary  code
$C^{(1)}=\{c\ (\text{mod}\ 2)\,|\,c\in C \}$, which is called the {\it residue code} of $C$, and the binary code $C^{(2)}=\{c\in \mathbb{Z}_2^n\,|\,2c\in C \}$, 
which is called the {\it torsion code} of $C$. 
If $C$ is a $\mathbb{Z}_4$-code of type $4^{k_1}2^{k_2}$, then $C^{(1)}$ is a binary code of dimension $k_1$ generated by the matrix
\[
\left[\begin{tabular}{ccc}
$I_{k_1}$&$A$&$B_1$\\
\end{tabular}\right].
\]
If $C$ is a self-dual $\mathbb{Z}_4$-code, then $C^{(1)}$ is doubly even and $C^{(1)}={C^{(2)}}^\bot$ (see \cite{ConS}).

According to \cite{survey}, the following statement holds.
\begin{thm}\label{Ctipa II}
Let $C$ be a Type II $\mathbb{Z}_4$-code of length $n.$ Then $C^{(2)}$ is an even binary code,
$C^{(1)}$  contains the all-ones binary vector, and $n\equiv 0\ (\text{mod}\ 8).$\\
\end{thm}

A {\it Boolean function} on $n$ variables is a mapping $f: \mathbb{F}_2^n \rightarrow \mathbb{F}_2.$ Its {\it  truth table} is the $(0, 1)$ sequence $(f((0,\dots ,0)),  f((0,\dots , 0,1)),\dots , f((1,\dots ,1))).$
The {\it  Walsh-Hadamard transformation} of  $f$ is
$$W_f (v) =\sum_{x\in \mathbb{F}_2^n} (-1)^{f(x)+\left\langle v, x\right\rangle}.$$
A  {\it  bent function} is a Boolean function $f$ such that  $W_f (v) =\pm 2^{\frac{n}{2}},$ for every $v\in \mathbb{F}_2^n.$ If $f$ is bent, then the number of its variables is an even number.
It was proven by Rothaus in 1960s  (\cite{Roth}) that the number of zeros of a bent function equals 
$$2^{n-1}\left( \pm \frac{1}{2^{\frac{n}{2}}}+1\right).$$

A {\it generalized Boolean function} on $n$ variables is a mapping $ f: \mathbb{F}_2^n \rightarrow \mathbb{Z}_{2^h}.$  The {\it generalized Walsh-Hadamard transformation} of  $f$ is
$$\tilde{f}(v) =\sum_{x\in \mathbb{F}_2^n}\omega^{f(x)} (-1)^{\left\langle v, x\right\rangle},$$
where $\omega=e^{\frac{2\pi i}{2^h}}.$
A  {\it generalized bent function (gbent function)} is a generalized Boolean function $f$ such that  $|\tilde{f} (v) |= 2^{\frac{n}{2}},$ for every $v\in \mathbb{F}_2^n.$ 
In this paper we will consider generalized bent functions from $\mathbb{F}_2^n$ into $\mathbb{Z}_4.$

\section{Codes constructed from gbent functions} \label{constr}
According to \cite{Sch}, the following theorem holds.

\begin{thm}\label{bentGbent}
	Let $m\geq 3$ be odd, and let $a, b: \mathbb{F}_2^{m-1}\rightarrow \mathbb{F}_2$ be bent functions. Then
 $f: \mathbb{F}_2^{m}\rightarrow \mathbb{Z}_4$ given by
$$f(x,y)=2a(x)(1+y)+2b(x)y+y,\  x\in  \mathbb{F}_2^{m-1}, y\in \mathbb{F}_2,$$
is a gbent function.
\end{thm}

\begin{lem}\label{cf}
Let $m\geq 3$ be odd and let $f :\mathbb{F}_2^{m}\rightarrow \mathbb{Z}_4$ be a gbent function constructed from bent functions $a$ and $b$ as in Theorem \ref{bentGbent}. 
Let  $c_f$ be a codeword $$(f((0,\dots ,0)),  f((0,\dots , 0,1)),\dots , f((1,\dots ,1)))\in \mathbb{Z}_4^{2^{m}}.$$
Then $wt_E(c_f)\equiv0\ (\text{mod}\ 8)$ and  $\langle c_f, c_f \rangle = 0$.
\end{lem}

\begin{proof}
By the construction, $c_f$ has $2^{m-1}$ even and $2^{m-1}$ odd coordinates. The number of zeros in $c_f$ is equal to the number of zeros in the bent function $a.$ This number is equal to $2^{m-2}\left( \pm \frac{1}{2^{\frac{m-1}{2}}}+1\right).$
It follows that $wt_E(c_f)=2^{m-1}+2^{m+1}\pm 2^{\frac{m+1}{2}}-2^{m}.$
So, $wt_E(c_f)$ is divisible by 8.
Since $wt_E(x)\equiv x_1^2+\dots +x_n^2\ (\text{mod}\ 8)$ for every $x\in \mathbb{Z}_4^n$, we have $\langle c_f, c_f \rangle = 0$.
\end{proof}

\begin{remark}
Let us notice that for $c_f=(f((0,\dots ,0)),  f((0,\dots , 0,1)),\dots , f((1,\dots ,1)))$ and $m=3$ Euclidean weight $wt_E(c_f)$ takes value $8$ or $16$.  
If $m\geq 5,$ it holds that $wt_E(c_f)\geq \frac{2^m}{3}+8$.
\end{remark}

\subsection{Codes over $\mathbb{Z}_4$}
An $n\times n$ {\it circulant matrix} is a matrix of the form 
$$\left[\begin{tabular}{ccccc}
$x_0$ & $x_{n-1}$ &\dots& $x_2$ & $x_1$\\
$x_{1}$ & $x_{0}$ & $x_{n-1}$ & \dots & $x_{2}$\\
\vdots  &   & & & \vdots \\
$x_{n-1}$ & \dots & \dots & $x_1$& $x_0$\\
\end{tabular}\right].$$

\begin{thm}\label{C_f}
Let $m\geq 3$ be odd, and let $a, b: \mathbb{F}_2^{m-1}\rightarrow \mathbb{F}_2$ be bent functions. Let  $f: \mathbb{F}_2^{m}\rightarrow \mathbb{Z}_4$ be a gbent function given by
$f(x,y)=2a(x)(1+y)+2b(x)y+y,\ $  $x\in  \mathbb{F}_2^{m-1},$ $y\in \mathbb{F}_2,$ and 
let  $c_f$ be a codeword $$(f((0,\dots ,0)),  f((0,\dots , 0,1)),\dots , f((1,\dots ,1)))\in \mathbb{Z}_4^{2^{m}}.$$
Let $C_f$ be a $\mathbb{Z}_4$-code generated by the $2^m \times 2^m$ circulant matrix whose first row is the codeword $c_f.$
Then $C_f$ is a self-orthogonal $\mathbb{Z}_4$-code of length $2^{m}$ and all its codewords have Euclidean weights divisible by $8.$ The residue code of $C_f$ has dimension $2.$
\end{thm}

\begin{proof}
By Lemma \ref{cf}, $\langle c_f, c_f \rangle=0$.

Let $c_i$ and $c_j$, $i \neq j$, be the $i$-th and the $j$-th row of the circulant generator matrix of $C_f$. If one of the indices is odd and the other is even, then
$$\left\langle c_i, c_j\right\rangle\equiv 0\cdot s_1+2\cdot s_2 \ (\text{mod}\ 4),$$
where $s_1$ is the sum of $2n_0(c_f)$ ones and threes, and $s_2$ is the sum of $2n_2(c_f)$ ones and threes. So, 
$s_2$ is an even number. It follows that $\langle c_i, c_j \rangle=0$.

If both of the indices are odd or both of them are even, then
$$\left\langle c_i, c_j\right\rangle\equiv \alpha_1 0\cdot 0 +\alpha_2 0\cdot 2 +\alpha_3 2\cdot 0+\alpha_42\cdot 2+\alpha_5 1\cdot 1+\alpha_6 1\cdot 3+\alpha_7 3\cdot 1+\alpha_8 3\cdot 3\ (\text{mod}\ 4),$$
where $\alpha_1+\alpha_2+\alpha_3+\alpha_4=\alpha_5+\alpha_6+\alpha_7+\alpha_8=2^{m-1}$ and $\alpha_6=\alpha_7.$
It follows that $$\left\langle c_i, c_j\right\rangle\equiv \alpha_5-2\alpha_6+\alpha_8 \ (\text{mod}\ 4)\equiv 2^{m-1}-4\alpha_6 \ (\text{mod}\ 4).$$ So, $c_i$ and $c_j$ are orthogonal codewords in this case as well.
Therefore, $C_f$ is a self-orthogonal $\mathbb{Z}_4$-code.

By Lemma \ref{cf}, $c_f$ has Euclidean weight divisible by 8. Since $C_f$ is a self-orthogonal code, all codewords of $C_f$ have Euclidean weights divisible by $8.$

It holds that $f(x,0)=2a(x)$ and $f(x,1)=2b(x)+1,$ for $x\in \mathbb{F}_2^{m-1}.$ Consequently,  even and odd coordinates alternate in the codeword $c_f$ and the residue code $C_f^{(1)}$ has dimension 2.

\end{proof}

\begin{exm} \label{ex1}
There are exactly eight bent functions on two variables. 
We constructed gbent functions $f$ from all pairs $(a,b)$ of bent functions $a, b: \mathbb{F}_2^2\rightarrow \mathbb{F}_2$, as given in Theorem \ref{bentGbent}. In that way, $64$ codewords $c_f\in\mathbb{Z}_4^8$ were obtained (see Lemma \ref{cf}). 
Among associated codes $C_f,$ constructed as in Theorem \ref{C_f}, there are two inequivalent codes.
One of them is the code $C_{f_{2^3\_1}}$
generated by the codeword  $c_{f_{2^3\_1}}=(0,1,0,1,0,3,2,1),$ which is obtained from the pair 
$(x_1x_2,x_1+x_1x_2).$
The other code
arises from the pair $(x_1x_2,x_1x_2),$ i.e., from the codeword $c_{f_{2^3\_2}}=(0,1,0,1,0,1,2,3)$. Both codes are self-orthogonal codes of type $4^22^3$ and their permutation automorphism groups have order $64.$
The codes obtained from the remaining 62 gbent functions on two variables are equal to $C_{f_{2^3\_1}}$ or to $C_{f_{2^3\_2}}.$ 

\end{exm}

\subsubsection{Type II codes over $\mathbb{Z}_4$}

\begin{thm}\label{prosC_f}
Let $m\geq 3$ be odd, and let $a, b: \mathbb{F}_2^{m-1}\rightarrow \mathbb{F}_2$ be bent functions. Let  $f: \mathbb{F}_2^{m}\rightarrow \mathbb{Z}_4$ be a gbent function given by
$f(x,y)=2a(x)(1+y)+2b(x)y+y,\ $  $x\in  \mathbb{F}_2^{m-1},$ $y\in \mathbb{F}_2,$ and 
let  $c_f$ be a codeword $$(f((0,\dots ,0)),  f((0,\dots , 0,1)),\dots , f((1,\dots ,1)))\in \mathbb{Z}_4^{2^{m}}.$$
Let $C_f$ be a cyclic $\mathbb{Z}_4$-code of type $4^22^{k_2}$ generated by $c_f.$
Let $G$ be a generator matrix of $C_f$ in standard form.  
Let $k_3=2^{m}-2^2-k_2$ and let
 $$\widetilde{D}=\left[\begin{tabular}{ccc}
$O$&$2I_{k_3}$&$H$\\

\end{tabular}\right]$$ be a $k_3\times 2^{m}$ matrix, where $O$ is the $k_3\times (k_2+2)$ null matrix and $H$ is a $k_3\times 2$ matrix whose rows $h_i, 1\leq i \leq k_3$ are defined as follows.

If $k_2$ is odd, then $$h_i=\left\{\begin{tabular}{cc}
$(0,2),$& if $i$ is odd\\
$(2,0),$& if $i$ is even\\
\end{tabular}\right..$$

If $k_2$ is even, then $$h_i=\left\{\begin{tabular}{cc}
$(2,0),$& if $i$ is odd\\
$(0,2),$& if $i$ is even\\
\end{tabular}\right..$$
\begin{enumerate}
\item[(i)] The code $\widetilde{C_f}$ generated by the matrix $\widetilde{G}=\left[\begin{tabular}{c}
$G$\\
$\widetilde{D}$\\
\end{tabular}\right]$
is a Type II $\mathbb{Z}_4$-code of length $2^{m}.$ 
\item[(ii)] If $m\geq 5,$ then $\widetilde{C_f}$ is a Type IV $\mathbb{Z}_4$-code.
\item[(iii)] Up to equivalence, $\widetilde{C_f}$ does not depend on the choice of bent functions $a$ and $b$.
\end{enumerate}
\end{thm}

\begin{proof}
 (i) By Theorem \ref{C_f}, $C_f$ is a self-orthogonal cyclic $\mathbb{Z}_4$-code generated by $c_f.$ It is of type $4^{2}2^{k_2}$ and length $2^{m},$ and all its codewords have Euclidean weights divisible by $8.$ 

The first and the second row of the matrix $G$, namely $g_1$ and $g_2$, are the only non-even rows in $G.$

Let $\widetilde{d}_i$ be the $i$-th row of the matrix $\widetilde{D}$.
Then $\left\langle g_1, \widetilde{d}_i\right\rangle=0$ and $\left\langle g_2, \widetilde{d}_i\right\rangle=0,$ for all $i=1,\dots,k_3$. 
Therefore, $\widetilde{C_f}$ is  a self-orthogonal $\mathbb{Z}_4$-code of type $4^22^{2^{m}-2^2}$, i.e., $\widetilde{C_f}$ is a self-dual $\mathbb{Z}_4$-code.

Moreover, all rows in $\widetilde{D}$ have Euclidean weight $8.$ It follows that Euclidean weights of all codewords in $\widetilde{C_f}$ are divisible by $8.$ 
We conclude that $\widetilde{C_f}$  is a Type II $\mathbb{Z}_4$-code of length $2^{m}.$

(ii) Let $m\geq 5$ and let $c\in \widetilde{C_f}.$ It holds $n_1(c)+n_3(c)\in \{0,2^{m-1},2^m\}.$ The Euclidean weight of $c$ is divisible by $8.$ So, $n_2(c)$ is an even number. 
It follows that $c$ has even Hamming weight. Therefore, $\widetilde{C_f}$  is a Type IV $\mathbb{Z}_4$-code for $m\geq 5.$ 

(iii) Let 
$\left[\begin{tabular}{cc}
$F$&$\tilde{I_2}$\\ 
\end{tabular}\right],$ where
$\tilde{I_2}=\left[\begin{tabular}{cc}
1&1\\
0&1\\
\end{tabular}\right]$ be the generator matrix of the residue code of $\widetilde{C_f}$.
 The number of Type II $\mathbb{Z}_4$-codes of type $4^22^{2^{m}-2^2}$ with the same residue code $\widetilde{C_f}^{(1)}$ is $2^2$ (see  \cite{Gaborit}, \cite{PlessLeonFields}). The generator matrices of those four codes could be given in the form (see \cite[Theorem 3]{PlessLeonFields})
 $$\left[\begin{tabular}{cc}
$F$&$\tilde{I_2}+2B$\\ 
$2H$&$O$\\ 
\end{tabular}\right],$$
where the possibilities for $B$ are
$\left[\begin{tabular}{cc}
0&1\\
1&0\\
\end{tabular}\right],$
$\left[\begin{tabular}{cc}
0&1\\
1&1\\
\end{tabular}\right],$
$\left[\begin{tabular}{cc}
1&1\\
1&0\\
\end{tabular}\right]$ and 
$\left[\begin{tabular}{cc}
1&1\\
1&1\\
\end{tabular}\right],$
if $m=3.$ 
If $m\geq 5,$ then the possibilities for $B$ are
$\left[\begin{tabular}{cc}
0&0\\
0&0\\
\end{tabular}\right],$
$\left[\begin{tabular}{cc}
0&0\\
0&1\\
\end{tabular}\right],$
$\left[\begin{tabular}{cc}
1&0\\
0&0\\
\end{tabular}\right]$ and 
$\left[\begin{tabular}{cc}
1&0\\
0&1\\
\end{tabular}\right].$
Therefore, those four codes are equivalent for every $m\geq3$.
 
\end{proof}

\begin{exm}\label{ex2}
The construction described in Theorem  \ref{prosC_f}, when applied on the codes $C_{f_{2^3\_1}}$ and $C_{f_{2^3\_2}}$ from Example \ref{ex1}, 
yields to a code equivalent to {\rsfs K}$_8$$'$, a unique Type II $\mathbb{Z}_4$-code of length $8$ and type $4^22^4$, whose permutation automorphism group has size $1152$ (see \cite{ConS}, \cite{clasZ4}). 
According to \cite{ConS}, it was introduced by Klemm in \cite{Klemm}. Let $W_i^E$ and $W_i^L$ denote the number of codewords of Euclidean weight $i$ and Lee weight $i$ in a $\mathbb{Z}_4$-code, respectively.
The code {\rsfs K}$_8$$'$ has Euclidean weight distribution
$$(W_0^E,W_8^E,W_{16}^E,W_{24}^E,W_{32}^E)=(1,140,102,12,1).$$
Its Lee weight distribution is $$(W_0^L,W_4^L,W_6^L,W_8^L,W_{10}^L,W_{12}^L,W_{16}^L)=(1,12,64,102,64,12,1).$$
\end{exm}

In the sequel we will consider weight distributions for codes $\widetilde{C_f}$ of length $2^m$ for odd $m, m\geq 3.$
By $A_i$ we denote the number of codewords of weight $i$ in a binary code.
It follows from the MacWilliams identity (see, e.g., \cite{FEC}, p. 252.) that the weight distribution $(A_0,\dots, A_n)$ of a binary linear $[n,k]$ code and the weight distribution $(A_0',\dots, A_n')$ of its dual code are connected by the equations
\begin{equation}\label{eq}
A_j'=\frac{1}{2^k}\sum_{i=0}^n A_i\sum_{l=0}^j (-1)^l {i \choose l}{n-i \choose j-l},\ j=0,\dots, n.
\end{equation}

\begin{lem} \label{torsion}
Let $\widetilde{C_f}$ be a Type II $\mathbb{Z}_4$-code of length $2^{m}$ for odd $m, m\geq 3,$ constructed as in Theorem \ref{prosC_f}. 
Then the weight distribution  of its torsion code $\widetilde{C_f}^{(2)}$ is $(A_0',\dots, A_{2^m}')$, where
$$A_j'=\frac{1}{2}\left({2^m \choose j}+\sum_{l=0}^j (-1)^l {2^{m-1} \choose l}{2^{m-1} \choose j-l}\right)$$ for even $j$ and $A_j'=0$ for odd $j$, $j=0,\dots, 2^m.$
\end{lem}
\begin{proof}
The $\mathbb{Z}_4$-code $\widetilde{C_f}$ constructed as in Theorem \ref{prosC_f} is a self-dual $\mathbb{Z}_4$-code. 
Therefore, the torsion code $\widetilde{C_f}^{(2)}$ is the dual code of $\widetilde{C_f}^{(1)}.$
Further, $\widetilde{C_f}$ is a Type II $\mathbb{Z}_4$-code. So, $\widetilde{C_f}^{(2)}$ is an even binary code.
By the construction, the residue code $\widetilde{C_f}^{(1)}$ contains codewords of weights $0, 2^{m-1}$ and $2^m$ with $A_0=1, A_{2^{m-1}}=2, A_{2^m}=1.$ 
The statement of the lemma follows from the expression (\ref{eq}). 
\end{proof}

\begin{thm} \label{weights}
 Let $\widetilde{C_f}$ be a Type II $\mathbb{Z}_4$-code of length $2^{m}$ for odd $m, m\geq 3,$ constructed as in Theorem \ref{prosC_f},
 and let $(A_0',\dots, A_{2^m}')$ be the weight distribution of its torsion code  $\widetilde{C_f}^{(2)}$.
 Then:
 \begin{enumerate}
\item[(i)] $\widetilde{C_f}$ has Euclidean weight distribution $(W_0^E,\dots,W_{2^{m+2}}^E)$ with $W_i^E=0$ for $i \not\equiv 0\ (mod\ 8)$ 
and, for $i$ divisible by $8,$ it holds $$W_i^E=A_{\frac{i}{4}}'+s_i+t_i,$$ 
\item[(ii)] the symmetrized weight enumerator of the code $\widetilde{C_f}$ is
 $$swe_{\widetilde{C_f}}(a,b,c)=s_{2^m}b^{2^m}+\sum_{i=0}^{2^m}\left(A_i'a^{2^m-i}c^{i}+t_{4i}a^{5\cdot 2^{m-3}-i}b^{2^{m-1}}c^{i-2^{m-3}}\right),$$
\item[(iii)]  if $m\geq 5,$ then $\widetilde{C_f}$ has Lee weight distribution $(W_0^L,\dots,W_{2^{m+1}}^L)$ with $W_i^L=0$ for $i \not\equiv 0\   (mod\ 4)$
and, for $i$ divisible by $4,$ it holds $$W_i^L=A_{\frac{i}{2}}'+s_i+u_i,$$
\end{enumerate}

where 
$$A_j'=\frac{1}{2}\left({2^m \choose j}+\sum_{l=0}^j (-1)^l {2^{m-1} \choose l}{2^{m-1} \choose j-l}\right)$$ for even $j$ and $A_j'=0$ for odd $j$, $j=0,\dots, 2^m,$ 
and

\begin{align*}
 s_i &=
  \begin{cases}
   2^{2^{m}-2},      & \text{if } i=2^m \\
   0,       & \text{otherwise}
  \end{cases},
  \\
 t_i &=
  \begin{cases}
   2^{2^{m-1}}{2^{m-1}\choose (2i-2^m)/8}, & \text{if } 2^{m-1}\leq i\leq 5\cdot2^{m-1} \\
   0,        & \text{otherwise}
  \end{cases},\\
 u_i &=
  \begin{cases}
   2^{2^{m-1}}{2^{m-1}\choose (2i-2^m)/4}, & \text{if } 2^{m-1}\leq i\leq 3\cdot2^{m-1} \\
   0,        & \text{otherwise}
  \end{cases}.
\end{align*}
\end{thm}

\begin{proof}  The expression for  $(A_0',\dots, A_{2^m}')$ is determined by Lemma \ref{torsion}.\\
By Theorem \ref{prosC_f},  $\widetilde{C_f}$ is a Type II $\mathbb{Z}_4$-code, i.e., all Euclidean weights are divisible by 8.
If $m\geq 5,$ then all Lee weights in $\widetilde{C_f}$ are divisible by four.\\
Let $\widetilde{G}_s$ be the generator matrix of $\widetilde{C_f}$ in standard form. Denote by $r_i$ the $i$-th row of $\widetilde{G}_s,$ $i=1,\dots,2^{m}-2.$
Note that the matrix $B_1+2B_2$ in (\ref{std}) is of order $2.$
Further, for $m=3$, each of the rows $r_1$ and $r_2$ contains the number $2$ exactly once, and for $m\geq 5,$ there are no $2$'s in $r_1$ and $r_2.$

Let $c\in \widetilde{C_f}.$ Then $n_1(c)+n_3(c)\in \{0,2^{m-1},2^m\}$ and
$$c=a_1r_1+a_2r_2+\sum_{i=3}^{2^{m}-2} a_ir_i,$$ where $a_1,a_2\in \mathbb{Z}_4$ and $a_i\in \mathbb{Z}_2$ for $i=3,\dots ,2^{m}-2.$ 
$\widetilde{C_f}$ contains exactly $2^{2^{m}-2}$ codewords $c$ with $n_1(c)+n_3(c)=0.$ These are the codewords with even $a_1$ and $a_2.$ 
Furthermore, $\widetilde{C_f}$ contains exactly $2^{2^{m}-2}$ codewords $c$ with $n_1(c)+n_3(c)=2^m.$ These are the codewords with odd $a_1$ and $a_2.$ 
Finally, $\widetilde{C_f}$ contains exactly $2^{2^m-1}$ codewords $c$ with $n_1(c)+n_3(c)=2^{m-1}.$ These are the codewords where one of the elements in $\{a_1, a_2\}$ is odd and the other is even.

The codewords $c$ with $n_1(c)+n_3(c)=2^m$ have Euclidean and Lee weight equal to $2^m.$
 
Let $c\in \widetilde{C_f}$ be a codeword with $n_1(c)+n_3(c)=2^{m-1}.$ If $m=3,$ half of these codewords have $n_2(c)=1$ and half of them have $n_2(c)=3.$ 
If $m\geq 5,$ then $n_2(c)$ is an even number, $n_2(c)\in \{0,\dots, 2^{m-1}\},$ 
and there are exactly $$2\cdot 2^{2^{m-1}-1}{2^{m-1} \choose n_2(c) }$$ codewords with Euclidean weight $2^{m-1}+4n_2(c)$ and Lee weight $2^{m-1}+2n_2(c)$, for each of these numbers $n_2(c).$ 
If $a_1$ is odd, then the even codeword $a_2r_2+\sum_{i=3}^{2^{m}-2} a_ir_i$ has $2$'s on exactly $n_2(c)$ even coordinate positions and the remaining even number of $2$'s are 
on odd coordinate positions. If $a_2$ is odd, then the even codeword $a_1r_1+\sum_{i=3}^{2^{m}-2} a_ir_i$ has $2$'s on exactly  $n_2(c)$ odd coordinate positions and the remaining even number of $2$'s are on even coordinate positions.

From these observations the weight distributions in (i) and (iii) are obtained.\\
For the coefficients of the symmetrized weight enumerator of $\widetilde{C_f},$ we count the codewords $c$ with $n_1(c)+n_3(c)=2^m,$ the even codewords and the codewords $c$ with $n_1(c)+n_3(c)=2^{m-1}$ in $\widetilde{C_f}.$

\end{proof}

\subsection{Binary Type II self-dual codes}

The {\it Gray map} $\phi : \mathbb{Z}_4^n\rightarrow \mathbb{F}_2^{2n}$ is the componentwise extension of the map
$\psi : \mathbb{Z}_4\rightarrow \mathbb{F}_2^{2}$ defined by
$\psi (0)=(0,0),\  \psi (1)=(0,1),\  \psi (2)=(1,1),\  \psi (3)=(1,0).$ 
Note that a $\mathbb{Z}_4$-code $C$ and the corresponding binary code $\phi (C)$ have the same size and that the Lee weight of a codeword $x\in \mathbb{Z}_4^n$ is equal to the (Hamming) weight of its Gray image $\phi(x)$.

If $C$ is a $\mathbb{Z}_4$-code of length $n,$ its Gray image $\phi (C)$ is a binary code of length $2n$, which is in general nonlinear. 
However, the following theorem holds (see \cite[Theorem 8]{latt}).

\begin{thm} \label{Lee4}
If $C$ is a self-dual $\mathbb{Z}_4$-code with all Lee weights divisible by 4, then the binary image of $C$ under the Gray map is linear.
\end{thm}

Moreover, according to \cite[Proposition 2.6]{typeIV}, the following statement holds.
 
\begin{thm} \label{LeeIV}
If $C$ is a Type IV $\mathbb{Z}_4$-code then all  the  Lee  weights  of $C$ are  divisible  by  four  and  
its Gray image is a self-dual doubly even binary code.
\end{thm}

If $x,y\in \mathbb{Z}_4^n,$ we define $xy$ as the componentwise product $(x_1y_1,\dots, x_ny_n).$ According to \cite{gray9}, the following statement holds.

\begin{lem}\label{lem1}
 
Let $C$ be a $\mathbb{Z}_4$-code of type $4^{k_1}2^{k_2}.$ Let $G$ be a generator matrix of $C$ in standard form and let $g_i$, $i \in \{1,2,...,k_1+k_2\},$ be its $i$-th row. The binary code $\phi(C)$ is linear if and only if for all $i,j\in \{1,\dots,k_1\}$ we have $2g_ig_j\in C.$
 
\end{lem}
 
As a consequence of our previous observations we have the following corollaries.
 
\begin{cor}\label{fiC_f}
 Let $\widetilde{C_f}$ be a Type II $\mathbb{Z}_4$-code of length $2^{m}$ for odd $m, m\geq 3,$ constructed as in Theorem \ref{prosC_f}. 
Then:
\begin{itemize}
\item[(i)] The Gray image $\phi(\widetilde{C_f})$ is a self-dual binary code of length $2^{m+1}.$ If $m\geq 5,$ then $\phi(\widetilde{C_f})$ is doubly even. 
\item[(ii)] The Gray image $\phi({C_f})$ is a self-orthogonal linear binary code of length $2^{m+1}.$  If $m\geq 5,$ then $\phi(C_f)$ is doubly even. 
\end{itemize}
 \end{cor}

\begin{proof}
\begin{itemize}
\item[(i)] According to Theorem \ref{prosC_f}, $\widetilde{C_f}$ is a Type II $\mathbb{Z}_4$-code of length $2^m$ and for $m\geq 5,$ $\widetilde{C_f}$ is a Type IV $\mathbb{Z}_4$-code.
Moreover, for $m=3$, a construction yields the code {\rsfs K}$_8$$'$ (see Examples \ref{ex1} and \ref{ex2}). 
Its Gray image is an even  $[16,8,4]$ binary code. Together with Theorem \ref{LeeIV} that concludes the proof.

\item[(ii)] The code  $\phi({C_f})$ is a subcode of the code $\phi(\widetilde{C_f})$. 
Let $G$ be a generator matrix of $C_f$ in standard form.
Codewords $g_1$ and $g_2$ have alternating odd and even coordinates. So, $g_1g_2$ is an even codeword. Then $2g_1g_2$ is a codeword with all coordinates equal to 0. 
It follows from Lemma \ref{lem1} that $\phi(C_f)$ is linear. Therefore, the statement holds.

\end{itemize}

\end{proof}

\begin{cor}
 Let $\widetilde{C_f}$ be a Type IV-II $\mathbb{Z}_4$-code of length $2^{m}$ for odd $m, m\geq 5,$ constructed as in Theorem \ref{prosC_f},  and let $(A_0',\dots, A_{2^m}')$ be the weight distribution of its torsion code  $\widetilde{C_f}^{(2)}$.
 Then the Gray image $\phi(\widetilde{C_f})$ has weight distribution $(W_0,\dots,W_{2^{m+1}})$ with $W_i=0$ for $i \not\equiv 0\ (mod\ 4)$ 
and, for $i$ divisible by $4,$ it holds $$W_i=A_{\frac{i}{2}}'+s_i+u_i,$$ where
\begin{align*}
 s_i &=
  \begin{cases}
   2^{2^{m}-2},      & \text{if } i=2^m \\
   0,       & \text{otherwise}
  \end{cases},
  \\
 u_i &=
  \begin{cases}
   2^{2^{m-1}}{2^{m-1}\choose (2i-2^m)/4}, & \text{if } 2^{m-1}\leq i\leq 3\cdot2^{m-1} \\
   0,        & \text{otherwise}
  \end{cases}.
\end{align*}
 \end{cor}
 
 \begin{proof}
 It follows directly from Theorem \ref{weights} (iii).
 \end{proof}

\section{Examples and related 1-designs} \label{struc}

In previous sections we constructed codes $C_f$, $\widetilde{C_f}$ and  $\phi(\widetilde{C_f})$ for $m=3$. In this section we give an example for $m=5$.
We also observe the minimum weight codewords and their relation with combinatorial designs.  
An incidence structure ${\mathcal D} =( {\mathcal P},{\mathcal B},{\mathcal I})$, with point set ${\mathcal P}$,
block set ${\mathcal B}$ and incidence ${\mathcal I}$ is a {\it $t$-$(v,k,\lambda)$ design}, if $|{\mathcal P}|=v$, every block
$B \in {\mathcal B}$ is incident with precisely $k$ points, and every $t$ distinct points are together incident with precisely $\lambda$ blocks. We assume that the reader is familiar with the basic facts of  design theory (see, e.g., \cite{Beth}, \cite{CRC}).\\

\underline{$m=3$}

In Examples \ref{ex1} and \ref{ex2} we constructed self-orthogonal $\mathbb{Z}_4$-codes $C_{f_{2^3\_1}}$, $C_{f_{2^3\_2}}$ of type $4^22^3$ and the corresponding Type II $\mathbb{Z}_4$-code $\widetilde{C_{f_{2^3}}}$ equivalent to {\rsfs K}$_8$$'$. The Gray image of $\widetilde{C_{f_{2^3}}}$ is a self-dual $[16,8,4]$ binary code. Here we observe the minimum weight codewords of the code $C_{f_{2^3\_1}}$ and related codes.

It holds $d_H(C_{f_{2^3\_1}})=2, d_L(C_{f_{2^3\_1}})=4$ and $d_E(C_{f_{2^3\_1}})=8,$ and the minimum Hamming weight codewords are the same as the  minimum Lee weight codewords. The supports of these codewords form a $1$-$(8,2,1)$ design, i.e., a resolvable $1$-design with $4$ blocks and the block intersection number $0.$ So, the minimum weight codewords of its Gray image form a resolvable $1$-$(16,4,1)$ design. 
The codewords of minimum Euclidean weight have Lee weight equal to $4, 6$ or $8.$ The supports of those codewords with Lee weight equal to $6$ form a $1$-$(8,5,5)$ design with $8$ blocks and block intersection numbers $2$ and $4.$  The supports of the codewords with Euclidean weight and Lee weight equal to $8$ form a trivial $1$-$(8,8,1)$ design with one block.

The dual code $C_{f_{2^3\_1}}^\bot$ is of type $4^{3}2^{3}$ with $d_H(C_{f_{2^3\_1}}^\bot)=2, d_L(C_{f_{2^3\_1}}^\bot)=d_E(C_{f_{2^3\_1}}^\bot)=4.$ 
In the code $C_{f_{2^3\_1}}^\bot$, the codewords of minimum Lee weight split in two classes: one class contains the codewords of minimum Hamming weight and the other contains the codewords of minimum Euclidean weight. 
The supports of the minimum Hamming weight codewords form a $1$-$(8,2,3)$ design with $12$ blocks and block intersection numbers $0$ and $1.$ 
The supports of the codewords of minimum Euclidean weight form a $1$-$(8,4,2)$ design with $4$ blocks and block intersection numbers $0$ and $2.$
That design is a $(2,2;2)$-net, i.e., an affine resolvable $1$-design.

It follows from Lemma \ref{torsion} that the torsion code $\widetilde{C_f}^{(2)}$ of  $\widetilde{C_f}$ has minimum weight $2.$ 
In the sequel, we observe the minimum weight codewords in the code $\widetilde{C_{f_{2^3}}}^{(2)}.$
In that code, the supports $\{i,j\},\ i<j,$ of the minimum weight codewords are divided in two classes:
$$\mathcal{M}(2)=\{\{i,j\}\ :\ j-i=2\ \text{or}\ j-i=6\}$$ and 
$$\mathcal{M}(4)=\{\{i,j\}\ :\ j-i=4\}.$$
The class $\mathcal{M}(2)$ consists of $8$ supports and every coordinate position occurs in exactly two supports in the class.
The class $\mathcal{M}(4)$ consists of $4$ supports and every coordinate position occurs in exactly one of the supports.
So, the set $\mathcal{M}(2)\cup \mathcal{M}(4)$ is the set of the blocks of a $1$-$(8,2,3)$ design with $12$ blocks.\\
The results for $m=3$ are summarized in Table \ref{m_3}.

\begin{table}[h]
	\centering  
		\begin{tabular}
	{|c|c|c|c|c|}	
\hline			
code&type/&$d_H,d_L,d_E$ &parameters,& block intersection \\
&parameters& &no. of blocks& numbers \\
\hline
$C_{f_{2^3\_1}}$&$4^22^3$&2,4,8&$1$-$(8,2,1)$, $4$ blocks&$0$    \\
            & & &$1$-$(8,5,5)^*$, $8$ blocks& $2,4$ \\
						&      &     &$1$-$(8,8,1)$, one block&    \\
\hline
	$\phi(C_{f_{2^3\_1}})$& $[16,7,4]$ & & $1$-$(16,4,1)$, $4$ blocks& $0$\\ 
\hline
$C_{f_{2^3\_1}}^\bot$&$4^32^3$&2,4,4&$1$-$(8,2,3)$, $12$ blocks& $0,1$\\
                 & & &$1$-$(8,4,2)$, $4$ blocks&$0,2$\\
\hline	
$\widetilde{C_{f_{2^3}}}^{(2)}$& $[8,6,2]$ & &$1$-$(8,2,3)$, $12$ blocks& $0,1$\\ 
\hline								
									
\end{tabular}
\caption{\it Constructed designs for $m=3$} \label{m_3}
\end{table}

\begin{remark}
 The block intersection graph $G_2$  of a $1$-$(8,5,5)$ design marked with $*$ in Table \ref{m_3} is a strongly regular graph with parameters $(8,4,0,4)$, i.e. the complete bipartite graph $K_{4,4}$.\\
\end{remark}

\underline{$m=5$}

Let $(a,b)=(x_1x_2+x_1x_3+x_2x_4,x_1x_2+x_3x_4)$ be a pair of bent functions $a, b: \mathbb{F}_2^4\rightarrow \mathbb{F}_2.$
From that pair, we constructed gbent function $f_{2^5}$ as described in Theorem \ref{bentGbent}. Further, the codeword $$c_{f_{2^5}}=(0, 1, 0, 1, 0, 1, 0, 3, 0, 1, 2, 1, 0, 1, 2, 3, 0, 1, 0, 1, 2, 1, 2, 3, 2, 3, 0, 3, 0, 3, 2, 1)$$
is constructed as in Lemma \ref{cf}.
The self-orthogonal $\mathbb{Z}_4$-code  $C_{f_{2^5}}$, constructed by Theorem \ref{C_f}, is of type $4^22^{21},$ its dual $C_{f_{2^5}}^\bot$ is of type $4^{9}2^{21}$ and its Gray image is a doubly even $[64,25,4]$ binary code. The permutation automorphism group of $C_{f_{2^5}}$ is of order $9663676416.$

$\mathbb{Z}_4$-code $C_{f_{2^5}}$ has $d_H(C_{f_{2^5}})=2, d_L(C_{f_{2^5}})=4$ and $d_E(C_{f_{2^5}})=8$, and the sets of minimum weight codewords are the same for all three weights. The supports of those codewords form a resolvable $1$-$(32,2,1)$ design with $16$ blocks and block intersection number $0.$ So, the minimum weight codewords of its Gray image yield a resolvable $1$-$(64,4,1)$ design. 

For $C_{f_{2^5}}^\bot$, it holds $d_H(C_{f_{2^5}}^\bot)=2, d_L(C_{f_{2^5}}^\bot)=4$ and $d_E(C_{f_{2^5}}^\bot)=8.$  The sets of minimum weight codewords are the same for Hamming and Lee weight. The supports of those codewords form a $1$-$(32,2,15)$ design with $240$ blocks and block intersection numbers $0$ and $1.$ 
The codewords of minimum Euclidean weight have Lee weight equal to $4, 6$ or $8.$ The supports of those codewords with Lee weight equal to $6$ form a $1$-$(32,5,20)$ design with $128$ blocks and block intersection numbers $0,1,2$ and $4.$ 
The supports of the codewords with Euclidean weight and Lee weight equal to $8$ form a  $1$-$(32,8,7)$ design with $28$ blocks and block intersection numbers $0$ and $4.$\\
Further, similarly to what we observed in the case $m=3$, the supports of minimum weight codewords in the torsion code $\widetilde{C_{f_{2^5}}}^{(2)}$ form a $1$-$(32,2,15)$ with $240$ blocks, as presented in Table \ref{m_5}.

\begin{table}[h]
	\centering  
		\begin{tabular}
	{|c|c|c|c|c|}	
\hline			
code&type/&$d_H,d_L,d_E$ &parameters,& block intersection \\
&parameters& &no. of blocks& numbers \\
\hline
$C_{f_{2^5}}$&$4^22^{21}$&2,4,8 &$1$-$(32,2,1)$, $16$ blocks&$0$\\
           
\hline
	$\phi(C_{f_{2^5}})$&$[64,25,4]$ & & $1$-$(64,4,1)$, $16$ blocks& $0$\\
\hline
$C_{f_{2^5}}^\bot$&$4^{9}2^{21}$&2,4,8 &$1$-$(32,2,15)$, $240$ blocks& $0,1$\\
                  & & &$1$-$(32,5,20)$, $128$ blocks& $0,1,2,4$  \\
						     &   &        &$1$-$(32,8,7)^*$, $28$ blocks& $0,4$    \\
\hline	
$\widetilde{C_{f_{2^5}}}^{(2)}$& $[32,30,2]$ & &$1$-$(32,2,15)$, $240$ blocks& $0,1$\\ 
\hline								
									
\end{tabular}
\caption{\it Constructed designs for $m=5$} \label{m_5}
\end{table}

\begin{remark}
 A $1$-$(32,8,7)$ design marked with $*$ in Table \ref{m_5} is a $(4,7;2)$-net, i.e., an affine resolvable $1$-design.
 Its block intersection graph $G_0$ is a strongly regular graph with parameters $(28,15,6,10)$.
 \end{remark}

Finally, observing the minimum weight codewords in the torsion code $\widetilde{C_f}^{(2)}$, for odd $m, m\geq 3,$  we obtain the following statement.
 
\begin{thm}
Let $\widetilde{C_f}$ be a Type II $\mathbb{Z}_4$-code of length $2^{m}$ for odd $m, m\geq 3,$ constructed as in Theorem \ref{prosC_f}. 
Then the supports of minimum weight codewords in its torsion code $\widetilde{C_f}^{(2)}$ form the set of the blocks of a $1$-$(2^m,2,2^{m-1}-1)$ design with $2^{m-1}(2^{m-1}-1)$ blocks.
\end{thm}
  
\begin{proof}
Let $\{i,j\},\ i<j,$ be a support of a minimum weight codeword in $\widetilde{C_f}^{(2)}.$ Then $j-i$ has to be an even number because $\widetilde{C_f}^{(2)}$ is the dual code of $\widetilde{C_f}^{(1)}.$ So,
the supports $\{i,j\}\subseteq \{1,...,2^m\},\ i<j,$ of the minimum weight codewords in $\widetilde{C_f}^{(2)}$ are divided into $\frac{2^m}{4}=2^{m-2}$ classes:
$$\mathcal{M}(k)=\{\{i,j\}\ :\ j-i=k\ \text{or}\ j-i=2^m-k\},$$ 
where $k\in\{2,4,\dots,2^{m-1}-2\},$ 
and 
$$\mathcal{M}\left(2^{m-1}\right)=\left\{\{i,j\}\ :\ j-i=2^{m-1}\right\}.$$
$\widetilde{C_f}^{(2)}$ is cyclic. So,
the class $\mathcal{M}(k),\ k\in\{2,4,\dots,2^{m-1}-2\},$ consists of $2^m$ supports and every coordinate position occurs in exactly two supports in the class.
The class $\mathcal{M}(2^{m-1})$ consists of $2^{m-1}$ supports and every coordinate position occurs in exactly one of the supports.
So, the set 
$$\mathcal{M}(2)\cup\dots \cup\mathcal{M}\left(2^{m-1}-2\right)\cup \mathcal{M}\left(2^{m-1}\right)$$
 is the set of the blocks of a $1$-$(2^m,2,2^{m-1}-1)$ design with $2^{m-1}(2^{m-1}-1)$ blocks.
\end{proof}

\begin{center}{\bf Acknowledgement}\end{center}
This work has been supported by {\rm C}roatian Science Foundation under the project 6732 and by the University of Rijeka under the project uniri-prirod-18-45.

\end{document}